\documentclass[11pt, letterpaper]{scrartcl}
\usepackage{amsmath,amsthm,amssymb,amsfonts}
\usepackage{mathtools}
\usepackage{mathrsfs}
\usepackage{multirow}
\usepackage[inline]{enumitem}
\usepackage{hyperref}
\usepackage{bm}
\usepackage[margin=1in]{geometry}
\usepackage{colortbl}
\usepackage{booktabs}
\usepackage{microtype}
\usepackage{xcolor}
\usepackage{accents}

\usepackage{setspace}
\usepackage{caption}
\usepackage{changes}
\makeatletter
\@namedef{Changes@AuthorColor}{red}
\colorlet{Changes@Color}{red}
\makeatother

\usepackage{natbib}
 \bibpunct[, ]{(}{)}{,}{a}{}{,}%

\allowdisplaybreaks

\newtheorem{assumption}{Assumption}
\newtheorem{theorem}{Theorem}
\newtheorem{corollary}{Corollary}
\newtheorem{lemma}{Lemma}
\newtheorem{proposition}{Proposition}
\theoremstyle{remark}

\newtheorem{definition}{Definition}

\DeclareMathOperator{\pr}{\mathbb P}
\DeclareMathOperator{\E}{\mathbb E}
\DeclareMathOperator{\ind}{\mathbb I}
\def\ud{\mathrm{d}}
\newcommand{\R}{\mathbb{R}}
\newcommand{\C}{\mathbb{C}}
\newcommand{\calX}{\mathcal{X}}
\newcommand{\diag}{\mathrm{diag}}
\newcommand{\ubar}[1]{\underaccent{\bar}{#1}}
\newcommand{\verti}[1]{{\left\vert #1 \right\vert}}
\newcommand{\vertii}[1]{{\left\vert\kern-0.25ex\left\vert #1 \right\vert\kern-0.25ex\right\vert}}
\newcommand{\vertiii}[1]{{\left\vert\kern-0.25ex\left\vert\kern-0.25ex\left\vert #1 \right\vert\kern-0.25ex\right\vert\kern-0.25ex\right\vert}}

\makeatletter
\renewcommand*{\@fnsymbol}[1]{\ensuremath{\ifcase#1\or \dagger\or *\or \ddagger\or
    \mathsection\or \mathparagraph\or \|\or **\or \dagger\dagger
    \or \ddagger\ddagger \else\@ctrerr\fi}}
\makeatother

\makeatletter
\renewenvironment{abstract}{%
\if@abstrt
    \small
    \begin{center}
      {\normalfont\sectfont\nobreak\abstractname
        \vspace{-.5em}\vspace{\z@}}%
    \end{center}
\fi
    \quotation
}{%
\endquotation
} 
\makeatother

\begin{document}

\title{\LARGE Affine Jump-Diffusions: Stochastic Stability and Limit Theorems}
\author{
{\large Xiaowei Zhang\thanks{Corresponding author. Department of Management Sciences, College of Business, City University of Hong Kong, Hong Kong. Email: \href{mailto:xiaowei.w.zhang@cityu.edu.hk}{xiaowei.w.zhang@cityu.edu.hk}}} 
\and 
{\large Peter W. Glynn\thanks{Department of Management Science and Engineering, Stanford University, CA 94305, U.S.} }
}

\date{}

\maketitle

\vspace{-60pt}
\begin{abstract}
\noindent
\paragraph{Abstract.} Affine jump-diffusions constitute a large class of continuous-time stochastic models that are particularly popular in finance and economics due to their analytical tractability. Methods for parameter estimation for such  processes require ergodicity in order establish consistency and
asymptotic normality of the associated estimators. In this paper, we develop stochastic stability conditions for affine jump-diffusions, thereby providing the needed large-sample theoretical support for estimating such processes. We establish ergodicity for such models by imposing a ``strong mean reversion'' condition and a mild condition on the distribution of the jumps, i.e. the finiteness of a logarithmic moment. Exponential ergodicity holds if the jumps have a finite moment of a positive order. In addition, we prove strong laws of large numbers and functional central limit theorems for additive functionals for this class of models. 
\end{abstract}

\section{Introduction}\label{sec:intro}
Affine jump-diffusion (AJD) processes constitute an important class of continuous time stochastic models that are widely used in finance and econometrics. This class of models is flexible enough to capture various empirical attributes such as stochastic volatility and leverage effects; see, e.g., \citet{Barndorff-NielsenShephard01}. Furthermore, the affine structure permits efficient computation, as a consequence of the fact that the characteristic function of its transient distribution is of an exponential affine form. The transform can then be computed by solving a system of ordinary differential equations (ODEs) of generalized Riccati type; see \citet{DuffiePanSingleton00}. The ability to efficiently compute such characteristic functions then leads to significant tractability both for computing various expectations and probabilities, and for the use of ``method of moments''  for calibrating such models; see \cite{Singleton01}, \cite{Bates06}, and \cite{FilipovicMayerhoferSchneider13}. 

AJD processes include a number of important special cases: the Ornstein-Uhlenbeck (OU) process, also known as the Vasicek model in \citet{Vasicek77}; the square-root diffusion process, also known as the Cox-Ingersoll-Ross (CIR) model in \citet{CIR1985}; and the Heston stochastic volatility model of \citet{Heston93}. Discussion of related extensions can be found in \citet{DuffieKan96,DaiSingleton00, DuffiePanSingleton00, Barndorff-NielsenShephard01,CheriditoFilipovicKimmel07}, and \citet{Collin-DufresneGoldsteinJones08}. Moreover, AJDs are also closely related to affine point processes, which involve jump intensities having an affine dependence on the state of the underlying AJD process. This class of point processes includes the Hawkes process \citep{Hawkes71} as a special case and has been widely used in financial economics to capture the self-exciting behavior that is often observed in event arrival data; see, e.g., \cite{ErraisGieseckeGoldberg10}, \cite{Ait-SahaliaCacho-DiazLaeven15}, and \cite{GaoZhouZhu18}.

As just noted, the ability to estimate and calibrate AJD models plays a key role in the popularity of such models in finance and economics. Of course, the large-sample theoretical support for estimation procedures for such AJD processes relies upon strong laws of large numbers (SLLNs), functional central limit theorems (FCLTs), and related ergodic theory. Our goal in this paper is to provide the first comprehensive development of the mathematical conditions under which such SLLNs and FCLTs hold for AJD processes. After stating our main results in Section \ref{sec:formulation}, we prove in Section \ref{sec:stability} that a canonical AJD is ergodic under mild conditions having to do with mean reversion and the distribution of the jumps. We then study conditions guaranteeing exponential ergodicity. Finally, in Section \ref{sec:limits} we develop SLLNs and FCLTs for AJDs. 

In terms of related literature, the study of large time ``moment explosions'' for AJD processes, and its close connection to implied volatility asymptotics, has been studied by, for example, \citet{Lee04}. A class of two-dimensional stochastic volatility models, of which the Heston model is a special case, is analyzed in \citet{AndersenPiterbarg07} to identify the ``explosion time'' at which the transient moment of a given positive order becomes infinite. \citet{GlassermanKim10} study the exponential moments of multidimensional affine diffusions (ADs) having no jumps. They characterize the domain of finiteness of the moment generating function of the diffusion process as well as the behavior of this domain in the ``large time'' limit. The existence of a non-degenerate limit of such moments implies tightness of the marginals of the AD, and is closely connected to existence of a stationary distribution. Their approach is based on the stability analysis of the limiting behavior of the Riccati equations that define the AD. Their results are extended in \citet{JenaKimXing12}, using a similar approach. See also \citet{Keller-Ressel11} for related analysis of a two-dimensional affine stochastic volatility model that permits jumps. Our work differs from the above papers primarily in two aspects. First, our model is a general AJD of arbitrary dimension, so that the models in the above papers are special cases of ours. In particular, the critical conditions for stochastic stability in these papers are special cases of ours. For instance, our stability condition is reduced to those in \citet{GlassermanKim10} and \citet{JenaKimXing12} in the absence of jumps; the stability condition of the so-called ``variance'' process in \citet{Keller-Ressel11} is a one-dimensional special case of ours. Second, their work focuses on existence of the limiting distribution, whereas we further establish ergodicity/exponential ergodicity results. The richer results on stochastic stability stem from the different approach we follow in this paper. Our analysis relies on Lyapunov criteria for Markov processes; see, e.g., \citet{MeynTweedie93_part3}. Related stability theory can be found in \citet{Masuda04,BarczyDoringLiPap14,JinRudigerTrabelsi16} and \cite{JinKremerRudiger17}, but the AJDs studied there all involve a state-independent jump intensity.

\section{Model Formulation and Main Results}\label{sec:formulation}

We will adopt the following notation throughout the paper. 
\begin{itemize}
\item
We write $\R_+^d:=\{v\in \R^d: v_i\geq 0, i=1,\ldots,d\}$ and $\R_-^d:=\{v\in \R^d: v_i\leq 0, i=1,\ldots,d\}$.
\item
A vector $v\in \R^d$ is treated as a column vector, $v^\intercal$ denotes its transpose, $\vertii{v}$ denotes its Euclidean norm.
\item
For a matrix $A$, $A\succeq 0$ means that $A$ is symmetric positive semidefinite and $A \succ 0$ means that $A$ is symmetric positive definite.
\item
We write $v_{\mathcal I}=(v_i:i\in {I}) $ and $A_{\mathcal{IJ}}=(A_{ij}:i\in I, j\in J)$, where $v\in\R^d$ is a vector, $A\in\R^{d\times d}$ is a matrix, and $\mathcal {I, J} \subseteq\{1,\ldots,d\}$ are two index sets. 
\item
We use $\mathbf{0}$ to denote a zero vector or a zero matrix, and $\mathrm{Id}(i)$ to denote a matrix with all zero entries except the $i$-th diagonal entry is 1, regardless of dimension.
\item
For a set $K\subseteq \R^d$, $\ind_K(x)$ denotes the indicator function associated with $K$, i.e. $\ind_K(x)=1$ if $x\in K$ and 0 otherwise.
\end{itemize}

Fix a complete probability space $(\Omega, \mathscr F, \pr)$ equipped with a filtration $\{\mathscr F_t:t\geq 0\}$ that satisfies the \emph{usual hypotheses} \citep[p.3]{Protter03}. Suppose that a stochastic process $X=(X(t):t\geq 0)$ with state space $\calX\subseteq \R^d$ satisfies the following stochastic differential equation (SDE) 
\begin{equation}\label{eq:SDE}
\begin{aligned}
\ud X(t) & =  \displaystyle\mu(X(t))\,\ud t+\sigma(X(t))\,\ud
W(t)+ \int_{\R^d}zN(\ud t,\ud z), \\
X(0) &= x\in\calX,
\end{aligned}
\end{equation}
where $W=(W(t):t\geq 0)$ is a $d$-dimensional Wiener process and $N(\ud t,\ud z)$ is a random counting measure on $[0,\infty)\times \R^d$ with compensator measure $\Lambda(X(t\text{-}))\ud t\nu(\ud z)$; moreover, $\mu:\R^d\mapsto\R^d$, $\sigma:\R^d\mapsto\R^{d\times d}$, and $\Lambda:\R^d\mapsto \R$ are measurable functions, and $\nu$ is a Borel measure on $\R^d$. In the sequel, we will write $\pr_x(\cdot)=\pr(\cdot|X(0)=x)$ and $\pr_\eta(\cdot)=\int_{\calX}\pr(\cdot|X(0)=x)\eta(\ud x)$ for an initial distribution $\eta$; $\E_x$ and $\E_\eta$ denote the corresponding expectation operators.

We call $X$ an AJD if the drift $\mu(x)$, diffusion matrix $\sigma(x)\sigma(x)^\intercal$, and jump intensity $\Lambda(x)$ are all affine in $x$, namely, 
\begin{equation}\label{eq:affine_coeff}
\begin{aligned}
\mu(x)& =b+\beta x, \qquad b\in \R^d,\; \beta\in\R^{d\times d}\\
\sigma(x)\sigma(x)^\intercal&= a+ \displaystyle\sum_{i=1}^dx_i\alpha_i,\qquad a\in\R^{d\times d},\; \alpha_i\in\R^{d\times d},\; i=1,\ldots, d \\
\Lambda(x)&=\lambda+\kappa^\intercal x, \qquad \lambda\in \R,\;\kappa\in\R^d.
\end{aligned}
\end{equation}

This paper is largely motivated by statistical calibration of AJDs. Most calibration procedures that have been applied to AJDs are based on some estimating equation as follows. Let $\Xi$ denote the collection of unknown parameters. For simplicity, we assume that the process $X$ is discretely sampled at time epochs $\{k\Delta:k=0,1,\ldots,n\}$ for some $\Delta>0$. To estimate $\Xi$, one judiciously selects a tractable function $h(x,y;\Xi)$ for which $\E[h(X(0),X(\Delta);\Xi)] =0$, and then solves the equation 
\[\frac{1}{n}\sum_{k=1}^n h(X((k-1)\Delta),X(k\Delta);\hat\Xi_n) = 0,\]
to compute the estimate $\hat\Xi_n$. In a situation where the dimension of $h$ is greater than the dimension of $\Xi$, one can use the generalized method of moments \citep{Hansen82}. Typical choices of $h$ include the marginal characteristic function of the conditional distribution of $X(k\Delta)$ given $X((k-1)\Delta)$ as in \citet{Singleton01}, or $\mathscr A g(x)$ for some tractable function $g$ with enough smoothness, where $\mathscr A$ is the operator defined in (\ref{eq:operator}) as in \citet{HansenScheinkman95}. See also \citet{DuffieGlynn04} for a choice of $h$ that also utilizes the operator $\mathscr A$ but in a context where $X$ is sampled at random times rather than deterministic times. 

In order to establish consistency and asymptotic normality of $\hat\Xi_n$, it is standard to assume positive Harris recurrence as well as certain moment conditions on the function $h$; see, e.g., \citet{Hansen82}. The SLLNs and FCLTs that we present as part of Theorem \ref{theo:state_ind_jump} and Theorem \ref{theo:state_dep_jump} provide large-sample theoretical support for establishing these asymptotic properties of the estimator. We refer interested readers to \citet{Ait-Sahalia07} for an extensive survey on various statistical calibration methods for general jump-diffusions and related assumptions for statistical validity.

\subsection{Main Assumptions}

The following three assumptions are universal throughout the paper.

\begin{assumption}\label{asp:state_space}
Let $\calX= \R_+^m\times \R^{d-m}$. For each $x\in\calX$, there exists a unique $\calX$-valued strong solution to the SDE \eqref{eq:SDE} with coefficients \eqref{eq:affine_coeff}.
\end{assumption}

\begin{assumption} \label{asp:admissible}
Let $ \mathcal  I=\{1,\ldots,m\}$ and $\mathcal  J=\{m+1,\ldots,d\}$ for some $0\leq m\leq d$. 
\begin{enumerate}[label=(\roman*)]
\item
$a\succeq0$ with $a_{\mathcal{II}}=\mathbf{0}$ 
\item
$\alpha_i\succeq0$ and $\alpha_{i,\mathcal{II}}=\alpha_{i,ii}\cdot\mathrm{Id}(i)$ for $i\in \mathcal I$; $\alpha_i=\mathbf{0}$ for $i\in  \mathcal J$;
\item
$b\in\R_+^m\times \R^{d-m}$;
\item
$\beta_{\mathcal {I J}}=\mathbf{0}$ and $\beta_{\mathcal {I I}}$ has non-negative off-diagonal elements;
\item
$\lambda\in\R_+$, $\kappa_{\mathcal I}\in\R^m_+$ and $\kappa_{\mathcal J}=\mathbf{0}$;
\item
$\nu$ is a probability distribution on $\calX $.
\end{enumerate}
\end{assumption}

\begin{assumption}\label{asp:irreducible}
$a_{\mathcal{JJ}}\succ 0$ and $2b_i>\alpha_{i, ii}>0$ for $i=1,\ldots,m$. 
\end{assumption}

In this paper, we focus on AJDs with \emph{canonical} state space (Assumption \ref{asp:state_space}) and \emph{admissible} parameters (Assumption \ref{asp:admissible}). In the absence of jumps (i.e., $\lambda=0$ and $\kappa=\mathbf{0}$), the existence and uniqueness of a strong solution to the SDE \eqref{eq:SDE} with coefficients \eqref{eq:affine_coeff} is established in \cite{FilipovicMayerhofer09}. They first prove the existence of a weak solution, then prove pathwise uniqueness of the solution, and finally apply the Yamada--Watanabe theorem \cite[Corollary 5.3.23]{KaratzasShreve91}. The same approach is followed in \cite{DawsonLi06} to prove the case of AJDs in one or two dimensions.

Clearly, under Assumption \ref{asp:admissible} both the diffusion matrix and the jump intensity are independent of $x_{\mathcal J}$, i.e. $\sigma(x)\sigma(x)^\intercal = a+\sum_{i=1}^m x_i\alpha_i$ and the jump intensity $\Lambda(x)=\lambda+\sum_{i=1}^mx_i\kappa_i$. In financial applications, the first $m$ components $(X_1,\ldots,X_m)$ are often used to model volatility processes and thus are referred to as \emph{volatility factors}, whereas the other $(d-m)$ components are referred to as \emph{dependent factors}. 

The jumps of the AJDs we study here have finite activity, a consequence of the fact that $\nu$ is assumed to be a probability distribution rather than a $\sigma$-finite measure. Nevertheless, this restriction is imposed merely for mathematical simplicity; the main results could also be proved for the case of infinite activity at the cost of a more involved analysis. One may recognize that the SDE \eqref{eq:SDE} with finite activity jumps is  precisely the model proposed in \citet{DuffiePanSingleton00}, which already covers a substantial number of financial and economic applications. 

For a one-dimensional AJD such as the CIR model, the condition $2b_i>\alpha_{i, ii}>0$ in Assumption \ref{asp:irreducible} is known as the Feller condition, which guarantees that the process stays positive. On the other hand, as detailed in \citet{MeynTweedie93_part2,MeynTweedie93_part3}, irreducibility is usually required in order to apply Lyapunov criteria to a continuous time Markov process. For instance, positive Harris recurrence is shown to be equivalent to ergodicity if some skeleton chain is irreducible in \citet{MeynTweedie93_part2}. Assumption \ref{asp:irreducible} serves this purpose (Proposition \ref{prop:irreducible}). Specifically, it is used to prove that $X$ admits a positive transition density. Note that the existence of a transition density for AJDs is established in \cite{FilipovicMayerhoferSchneider13} but their proof requires $b_i>\alpha_{i, ii}>0$ for $i=1,\ldots,m$, which is stronger than our Assumption \ref{asp:irreducible}.

\subsection{Main Results}
Prior to presenting the main results of the paper, let us review several concepts regarding stochastic stability of a Markov process.  

\begin{definition}
A Markov process $X$ with state space $\calX$ is called \emph{Harris recurrent} if there exists a non-trivial $\sigma$-finite measure $\varphi$ on $\calX$ such that $\int_0^\infty \ind_K(X(t))\,\ud t = \infty$, $\pr_x$-a.s., for all $x\in\calX$ and any measurable set $K$ with $\varphi(K)>0$.
\end{definition}

\begin{definition}
A Harris recurrent Markov process $X$ is called \emph{positive Harris recurrent} if it admits a finite invariant measure $\pi$, which can be normalized to a probability measure that is called the \emph{stationary distribution} of $X$, the measure $\pi$ must necessarily be unique.
\end{definition}

\begin{definition}
For a Markov process $X$ with state space $\calX$, a set $K\subseteq \calX$ is called \emph{uniformly transient} if there exists $M<\infty$ such that $\E_x\int_0^\infty \ind_K(X(t))\,\ud t\leq M$ for all $x\in\calX$. Furthermore, $X$ is called \emph{transient} if there is a countable cover of $\calX$ with uniformly transient sets.
\end{definition}

\begin{definition}\label{def:norm}
For any measurable function $f:\calX\mapsto [1,\infty)$ and any signed-measure $\varphi$ on $\calX$, define the \emph{$f$-norm} of $\varphi$ by $\vertii{\varphi}_f\coloneqq \sup_{|h|\leq f} |\varphi(h)|$, where
$\varphi(h)\coloneqq \int_{\calX} h(x)\, \varphi(\ud x)$.
When $f\equiv 1$, $\vertii{\cdot}_f$ is called the \emph{total variation
  norm} and is denoted by $\vertii{\cdot}$.
\end{definition}

The following concept is also needed to state our condition for stochastic stability for AJDs.
\begin{definition}
A square matrix is called \emph{stable} if all its eigenvalues have negative real parts.
\end{definition}

The following notation will facilitate the presentation in the sequel. Let $Z$ denote an $\R^d$-valued random variable with distribution $\nu$. For $q>0$, set $f_q(x)\coloneqq 1+\vertii{x}^q$. For any measurable functions $f:\calX\mapsto [1,\infty)$ and $h:\calX\mapsto \R$, set 
$\vertii{h}_f\coloneqq \sup_{x\in\calX}\{\verti{h(x)}/f(x)\}$.
Let $\mathcal D[0,1]$ denote the space of right continuous functions $x:[0,1]\mapsto \R$ with left limits, endowed with the Skorokhod topology. 

A distinctive feature of AJDs, besides the affine structure, relative to other jump-diffusion models is that its jump intensity is state-dependent. This property endows AJDs with greater flexibility in financial modeling but creates technical difficulties for analyzing the dynamics of the process. 
Indeed, differing theoretical treatments are needed, depending on whether the jump intensity is state-dependent, when we establish Lyapunov inequalities in Section \ref{sec:stability}. We therefore present our main results in two separate theorems. Theorem \ref{theo:state_ind_jump} covers only AJDs with state-independent jump intensities ($\kappa=\mathbf{0}$), whereas Theorem \ref{theo:state_dep_jump} allows state-dependent jump intensities. 

\begin{theorem}
\label{theo:state_ind_jump}
If Assumptions \ref{asp:state_space}--\ref{asp:irreducible} hold, $\kappa=\mathbf{0}$, $\beta$ is a stable matrix, and $\E\log(1+\vertii{Z})<\infty$, then:
\begin{enumerate}[resume,label=(\roman*)]
\item\label{item:ergo_ind}
$X$ is positive Harris recurrent and 
\begin{equation}\label{eq:ergodic_convergence}
\lim_{t\to\infty}\vertii{ \pr_x(X(t)\in\cdot) - \pi(\cdot)}=0,\quad x\in\calX,
\end{equation}
where $\pi$ is the stationary distribution of $X$.
\end{enumerate}
If, in addition, $\E\vertii{Z}^p<\infty$ for some $p>0$, then:
\begin{enumerate}[resume,label=(\roman*)]
\item\label{item:exp_ergo_ind}
For each $q\in(0,p]$, there exist positive finite constants $c_q$ and $\rho_q$ such that 
\begin{equation}\label{eq:exp_ergodic_convergence}
\vertii{\pr_x(X(t)\in\cdot) - \pi(\cdot)}_{f_q}\leq c_q f_q(x)e^{-\rho_qt},\quad t\geq 0,\;x\in\calX.
\end{equation}
\item\label{item:SLLN_ind}
For any measurable function $h:\calX\mapsto \R$ with $\vertii{h}_{f_p}<\infty$, 
\begin{equation}
\label{eq:LLN_cont_func}
\pr_x\left(\lim_{t\to\infty}\frac{1}{t}\int_0^t h(X(s))\,\ud s = \pi(h) \right)=1, \quad x\in\calX,
\end{equation}
and
\begin{equation}
\label{eq:LLN_disc_func}
\pr_x\left(\lim_{n\to\infty}\frac{1}{n}\sum_{i=1}^n h(X(i\Delta)) = \pi(h) \right)=1, \quad x\in\calX.
\end{equation}
\item\label{item:FCLT_ind}
For any measurable function $h:\calX\mapsto \R$ with $\vertii{h^q}_{f_p}<\infty$ for some $q>2$, there exist non-negative finite constants $\sigma_h$ and $\gamma_h$ such that 
\begin{equation}
\label{eq:FCLT_cont_func}
n^{1/2}\left(\frac{1}{n}\int_0^{n\cdot} h(X(s))\,\ud s - \pi(h)\right)\Rightarrow \sigma_h W(\cdot),
\end{equation}
and 
\begin{equation}
\label{eq:FCLT_disc_func}
n^{1/2}\left(\frac{1}{n}\sum_{i=1}^{\lfloor n\cdot\rfloor} h(X(i\Delta))-\pi(h)\right)\Rightarrow \gamma_h W(\cdot),
\end{equation}
as $n\to\infty$ $\pr_x$-weakly in $\mathcal D[0,1]$ for all $x\in\calX$, where $W$ is a one-dimensional Wiener process.
\end{enumerate}
\end{theorem}

\begin{theorem}
\label{theo:state_dep_jump}
If Assumptions \ref{asp:state_space}--\ref{asp:irreducible} hold, $\beta+\E(Z)\kappa^\intercal$ is a stable matrix, and $\E\vertii{Z}<\infty$, then:
\begin{enumerate}[resume,label=(\roman*)]
\item\label{item:ergo_dep}
$X$ is positive Harris recurrent and (\ref{eq:ergodic_convergence}) holds.
\end{enumerate}
If, in addition, $\E\vertii{Z}^p<\infty$ for some $p\geq 1$. Then:
\begin{enumerate}[resume,label=(\roman*)]
\item\label{item:exp_ergo_dep}
For each $q\in[1,p]$, there exist positive finite constants $c_q$ and $\rho_q$ such that (\ref{eq:exp_ergodic_convergence}) holds.
\item\label{item:SLLN_dep}
For any measurable function $h:\calX\mapsto \R$ with $\vertii{h}_{f_p}<\infty$, (\ref{eq:LLN_cont_func}) and (\ref{eq:LLN_disc_func}) hold.
\item\label{item:FCLT_dep}
For any measurable function $h:\calX\mapsto \R$ with $\vertii{h^q}_{f_p}<\infty$ for some $q>2$, there exist non-negative finite constants $\sigma_h$ and $\gamma_h$ such that (\ref{eq:FCLT_cont_func}) and (\ref{eq:FCLT_disc_func}) hold 
as $n\to\infty$ $\pr_x$-weakly in $\mathcal D[0,1]$ for all $x\in\calX$.
\end{enumerate}
\end{theorem}

We note that $X$ is called \emph{ergodic} if it has a stationary distribution $\pi$ and the convergence (\ref{eq:ergodic_convergence}) holds, whereas called $f$-\emph{exponentially ergodic} if
\[\vertii{\pr_x(X(t)\in\cdot) - \pi(\cdot)}_f\leq c(x)e^{-\rho t},\quad t\geq 0,\;x\in\calX,\]
for some functions $f:\calX\mapsto [1,\infty)$, $c:\calX\mapsto \R_+$ and some positive finite constant $\rho$. Clearly, $X$ is $f_p$-exponentially ergodic under the assumptions of Theorem \ref{theo:state_ind_jump}\ref{item:exp_ergo_ind} or Theorem \ref{theo:state_dep_jump}\ref{item:exp_ergo_dep}.

The key condition imposed here to establish positive Harris recurrence of AJDs is that $\beta+\E(Z)\kappa^\intercal$ is a stable matrix. If we adopt the convention that $0\times \infty=0$, when $\kappa=\mathbf{0}$ this condition is reduced to that $\beta$  is a stable matrix regardless of the finiteness of $\E\vertii{Z}$. The condition that $\beta$ is a stable matrix is typically assumed in the literature, including \citet{SatoYamazato84}, \citet{GlassermanKim10}, and \citet{JenaKimXing12},  in order that the process be mean reverting and have a stationary distribution. However, the first of the three articles works on a special L{\'e}vy-driven SDE, whereas the other two study ADs, so none of them allows state-dependent jump intensities as AJDs do. It can be shown that the stability of $\beta + \E(Z)\kappa^\intercal$ implies that of $\beta$; see Lemma 3 of \cite{ZhangBlanchetGieseckeGlynn15}. Thus, our condition is  stronger and we call it the \emph{strong mean reversion condition}. 

Note that $\E(Z)$ is the mean jump size and  $\kappa$ largely determines the magnitude of the jump intensity when the AJD takes on big values. To some extent, $\E(Z)\kappa^\intercal$ captures the impact of the jumps. Thus, by imposing the stability of $\beta+\E(Z)\kappa^\intercal$, we essentially assume that mean reversion is a dominating factor, more significant than the jumps, in the dynamics of the process. On the other hand, this condition is technically mild. Indeed, we show in Section \ref{sec:remarks} that it cannot be relaxed in general if positive Harris recurrence of an AJD is desired.

\section{Stochastic Stability}\label{sec:stability}

In this section, we apply  Lyapunov ideas to address the stochastic stability of $X$. A key step in this approach is to judiciously construct suitable Lyapunov functions; see \citet{MeynTweedie93_part3} for an extensive treatment of this approach. Nevertheless, we do not directly use the results there because their theory uses a definition of domain that insists on functions inducing martingales, whereas we work with local martingales.

Consider a twice-differentiable function $g:\calX\mapsto \R$. By virtue of It\^o's formula, 
\begin{align}
g(X(t)) & = g(X(0)) + \int_0^t \bigg[\nabla g(X(s\text{-}))^\intercal \mu(X(s\text{-})) + \frac{1}{2}\sum_{i,j=1}^d\frac{\partial^2 g(X(s\text{-}))}{\partial x_i\partial x_j}(\sigma(X(s\text{-}))\sigma(X(s\text{-}))^\intercal)_{ij}\bigg]\ud s \nonumber \\
&+ \int_0^t \nabla g(X(s\text{-}))^\intercal \sigma(X(s\text{-}))\,\ud W(s) + \int_0^t \int_{\calX} (g(X(s\text{-})+z) - g(X(s\text{-}))) N(\ud s,\ud z). \label{eq:ito}
\end{align}
By defining operators $\mathscr G$, $\mathscr L$, and $\mathscr A$ on twice-differentiable appropriately integrable functions $g$ via 
\begin{equation}
\label{eq:operator}
\begin{aligned}
\mathscr G g(x) \coloneqq & \nabla g(x)\cdot (b+\beta x) + \frac{1}{2}\sum_{i,j=1}^ d \frac{\partial ^2g(x)}{\partial x_i\partial   x_j}\left(a_{i,j}+\sum_{k=1}^d \alpha_{k,ij} x_k\right), \\
\mathscr L g(x) \coloneqq &  (\lambda+\kappa^\intercal x) \int_{\calX}(g(x+z)-g(x))\,\nu(\ud z), \\[0.5ex]
\mathscr A g(x) \coloneqq & \mathscr G g(x) + \mathscr L g(x),
\end{aligned}
\end{equation}
we may rewrite (\ref{eq:ito}) as 
\begin{equation}
\label{eq:ito_2}
\begin{aligned}
g(X(t)) =&  g(X(0)) + \int_0^t \mathscr Ag(X(s\text{-}))\,\ud s + S_1(t) + S_2(t),\\ 
S_1(t)\coloneqq &  \int_0^t \nabla g(X(s\text{-}))^\intercal \sigma(X(s\text{-}))\,\ud W(s), \\ 
S_2(t)\coloneqq& \int_0^t \int_{\calX} (g(X(s\text{-})+z) - g(X(s\text{-}))) \tilde N(\ud s,\ud z),
\end{aligned}
\end{equation}
where $\tilde N(\ud s,\ud z) = N(\ud s,\ud z)-\Lambda(X(s\text{-}))\ud s\nu(\ud z)$ is the compensated random measure of $ N(\ud s,\ud z)$. 

We introduce some notation to facilitate the construction of the needed Lyapunov inequalities. First, for a $d\times d$ matrix $H\succ 0$, define $\vertii{v}_H \coloneqq \sqrt{v^\intercal H v}$. Then, $\vertii{\cdot}_H$ is a \emph{vector norm} on $\R^d$ and it is easy to show that 
\begin{equation}\label{eq:norm_equivalence}
\ubar{\delta} \vertii{v}^2 \leq \vertii{v}_H^2\leq \bar{\delta}\vertii{v}^2,\quad v\in\R^d,
\end{equation}
where $(\delta_i:i=1,\ldots,d)$ are the eigenvalues of $H$, $\ubar{\delta} = \min\{\delta_i:i=1,\ldots,d\}$ and $\bar{\delta} = \max\{\delta_i:i=1,\ldots,d\}$. We can then define the following \emph{induced} matrix norms (see \citet[p.340]{HornJohnson12}). For a matrix $A\in\R^{d\times d}$, define 
\[\vertiii{A}\coloneqq\sup \left\{\frac{\vertii{Av}}{\vertii{v}}: \mathbf 0 \neq v\in\R^d\right\}\quad 
\mbox{and}\quad 
\vertiii{A}_H\coloneqq\sup \left\{\frac{\vertii{Av}_H}{\vertii{v}_H}: \mathbf 0 \neq v\in\R^d\right\}.
\]

\subsection{Positive Harris Recurrence and Ergodicity}\label{sec:ergodicity}

For each $\Delta>0$, let $X^\Delta\coloneqq(X(n\Delta):n=0,1,\ldots)$ denote the $\Delta$-skeleton of $X$. 

\begin{proposition}
\label{prop:irreducible}
Under Assumptions \ref{asp:state_space}--\ref{asp:irreducible}, $X^\Delta$ is $\varphi$-irreducible for any $\Delta>0$, where $\varphi$ is the Lebesgue measure on $\calX$. 
\end{proposition}

The proof of Proposition \ref{prop:irreducible} relies on the following result, which is of interest in its own right. It reduces irreducibility of a jump-diffusion process to that of the associated diffusion process. 

\begin{lemma}
\label{lemma:irreducible}
Suppose that $X$ satisfies the SDE \eqref{eq:SDE}\footnote{Here, we do not restrict its coefficients $\mu$, $\sigma$, $\Lambda$ to follow the affine form (\ref{eq:affine_coeff}).}. Let $\tilde X=(\tilde X(t):t\geq 0)$ satisfy 
\begin{equation}\label{eq:AD}
\begin{aligned}
\ud \tilde X(t) & =  \displaystyle\mu(\tilde X(t))\,\ud t+\sigma(\tilde X(t))\,\ud
W(t), \\
\tilde X(0) &= x\in\calX,
\end{aligned}
\end{equation}
where $W$ is the $d$-dimensional Wiener process in \eqref{eq:SDE}. If $\tilde X^\Delta$ (resp., $\tilde X$) is $\varphi$-irreducible, then $X^\Delta$ (resp., $X$) is $\varphi$-irreducible.
\end{lemma}
\begin{proof}
Consider a measurable $K\subseteq \calX$ and let $\tau$ denote the first jump time of $X$. Then $\pr_x(X(t) = \tilde X(t)) = 1$ for $t<\tau^*$. It follows that for any $t>0$,
\begin{align*}
\pr_x(X(t)\in K, \tau^*>t) & = \E_x \left[\E\left(\ind(\tilde X(t)\in K, \tau^*>t)|X(s),0\leq s\leq t\right)\right] \\[0.5ex]
& = \E_x \left[\ind(\tilde X(t)\in K) \pr \left(\tau^*>t|\tilde X(s),0\leq s\leq t\right)\right] \\[0.5ex]
& = \E_x \left[\ind(\tilde X(t)\in K) e^{-\int_0^t \Lambda(\tilde X(s))\,\ud s}\right].
\end{align*}
Hence, $\pr_x(X(t)\in K, \tau^*>t)=0$ if and only if $\pr_x(\tilde X(t)\in K)=0$ for any $t>0$. It is then clear that the $\varphi$-irreducibility of  $\tilde X^\Delta$ (resp., $\tilde X$) implies that of $X^\Delta$ (resp., $X$). 
\end{proof}

\begin{proof}[Proof of Proposition \ref{prop:irreducible}.]

The key in the proof is to convert the AJD by a linear transformation used in \citet{FilipovicMayerhofer09} into a canonical representation in which the matrices involved are of special form. Specifically, note that if $X$ satisfies the SDE \eqref{eq:SDE} with coefficients (\ref{eq:affine_coeff}), then for any nonsingular matrix $A\in\R^{d\times d}$, the linear transformation $Y=AX$ satisfies 
\begin{equation}\label{eq:SDE_linear_trans}
\begin{aligned}
\ud Y(t) & =  (Ab + A\beta A^{-1}Y(t))\, \ud t + A\sigma\left(A^{-1}Y(t)\right)\,\ud W(t) + \displaystyle\int_{\R^d}Az N(\ud t,\ud z), \\
Y(0) & = A x,
\end{aligned}
\end{equation}
where $N(\ud t,\ud z)$ has the compensator measure $\Lambda(A^{-1}Y(t\text{-}))\ud t\nu(\ud z)$. So the drift, diffusion matrix, and intensity of SDE (\ref{eq:SDE_linear_trans}) are
$Ab + A\beta A^{-1}y$,  $A\sigma\left(A^{-1}y\right)\sigma\left(A^{-1}y\right)^\intercal A^\intercal$, and $\lambda+\kappa^\intercal A^{-1}y$, 
respectively, which are all affine in $y$. Consequently, the existence and uniqueness of a strong solution to \eqref{eq:SDE} is invariant with respect to nonsingular linear transformations. 

Since $\alpha_{i,ii}>0$ for all $i=1,\ldots,m$, it follows from Lemma 7.1 of \citet{FilipovicMayerhofer09} that there exists a nonsingular matrix $A\in\R^{d\times d}$ that maps $\R_+^m\times \R^{d-m}$ to itself and renders the transformed diffusion matrix in the following block-diagonal form 
\[
A\sigma\left(A^{-1}y\right)\sigma\left(A^{-1}y\right)^\intercal A^\intercal = \begin{pmatrix}
\diag(\alpha_{1,11}y_1,\ldots,\alpha_{m,mm}y_m) & \mathbf{0} \\
\mathbf{0} & h + \sum_{i=1}^m y_i \eta_i
\end{pmatrix}
\]
for some $(d-m)\times (d-m)$ matrices $h\succeq 0$ and $\eta_i\succeq 0$, $i=1,\ldots,m$. In particular, $A$ is of the form
\[A = 
\begin{pmatrix}
I_m & \mathbf{0} \\
D & I_{d-m} 
\end{pmatrix},
\]
for some $(d-m)\times m$ matrix $D$, where $I_m$ and $I_{d-m}$ are identity matrices. Moreover, it is straightforward to verify that $Ab$, $A\beta A^{-1}$, and $\kappa^\intercal A^{-1}$ satisfy both Assumption \ref{asp:admissible} and Assumption \ref{asp:irreducible} in lieu of $b$, $\beta$, and $\kappa$. Hence, we can assume without loss of generality  that the diffusion matrix of \eqref{eq:SDE} has the form 
\begin{equation}\label{eq:linear_trans}
\sigma(x)\sigma(x)^\intercal  = 
 \begin{pmatrix}
\diag(\alpha_{1,11}x_1,\ldots,\alpha_{m,mm}x_m) & \mathbf{0} \\
\mathbf{0} & a_{\mathcal{JJ}} + \sum_{i=1}^m x_i \alpha_{i,\mathcal{JJ}}
\end{pmatrix}.
\end{equation}
Hence, $\tilde X_{\mathcal I}(t)$ satisfies
\[
\begin{aligned}
\ud \tilde X_{\mathcal I}(t) & = (b_{\mathcal I} + \beta_{\mathcal{II}} \tilde X_{\mathcal I}(t))\,\ud t + 
\diag(\sqrt{\alpha_{1,11}x_1},\ldots,\sqrt{\alpha_{m,mm}x_m})\,
\ud W_{\mathcal I}(t), \\ 
\tilde X_{\mathcal I}(0) & =  x_{\mathcal I}\in\R_+^m.
\end{aligned}
\]
With the assumption that $2b_i>\alpha_{i,ii}$, $i=1,\ldots,m$, we can directly verify the conditions of the theorem on p.388 of \citet{DuffieKan96} to conclude that $\mathbf{0}\in\R_+^m$ is not attainable in finite time, i.e. $\tilde X_i(t)>0$ for all $t>0$ and $i=1,\ldots,m$, if $\tilde X_i(0)>0$, $i=1,\ldots,m$.

We now consider a bijective transformation $\tilde Y\coloneqq f(\tilde X)$, where $f:\calX\mapsto \calX$ is defined as follows: $f_i(x)=2\sqrt{x_i}$ for $i=1,\ldots,m$ and $f_i(x)=x_i$ for $x=m+1,\ldots,d$. Then, 
\[\frac{\partial f_i(x)}{\partial x_j}= \left\{
\begin{array}{ll}
\displaystyle x_i^{-1/2}, & \mbox{if } i=j, i=1,\ldots,m, \\
1, & \mbox{if } i=j, i=m+1,\ldots,d, \\
0, & \mbox{otherwise,}
\end{array}
\right.
\]
and 
\[ \frac{\partial^2 f_i(x)}{\partial x_k\partial x_l} = \left\{
\begin{array}{ll}
-\frac{1}{2}x_i^{-3/2}, & \mbox{if } i=k=l, i=1,\ldots,m, \\
0, & \mbox{otherwise.}
\end{array}
\right.
\]
It follows that, by It\^o's formula,
\[\ud f_i(\tilde X(t)) = \zeta_i(\tilde X(t))\,\ud t + \nabla f_i(\tilde X(t))^\intercal \sigma(\tilde X(t))\,\ud W(t), \]
for $i=1,\ldots,d$, where 
\[\zeta_i(x) = \frac{\partial f_i(x)}{\partial x_i} \mu_i(x) + \frac{1}{2} \frac{\partial^2 f_i(x)}{\partial x_i^2}  (\sigma(x)\sigma(x)^\intercal )_{ii}.\]
Note that we have shown that $x_i>0$, $i=1,\ldots,m$ for $x\in\calX$, so the function $\zeta(x)$ is well-defined for all $x\in\calX$. Let $f^{-1}$ denote the inverse mapping of $f$, i.e. $f^{-1}_i(y) = y_i^2$ for $i=1,\ldots,m$ and $f^{-1}_i(y)=y_i$, for $i=m+1,\ldots,d$. Then, 
\begin{equation}
\label{eq:SDE_Y_tilde}
\ud \tilde Y(t) = \zeta(f^{-1}(\tilde Y(t)))\,\ud t + \nabla f(f^{-1}(\tilde Y(t)))\sigma(f^{-1}(\tilde Y(t)))\,\ud W(t),
\end{equation}
where $ \nabla f\coloneqq (\frac{\partial f_i}{\partial x_j})_{1\leq i,j\leq d} $ is the Jacobian matrix of $f$. A straightforward calculation reveals that the diffusion matrix of (\ref{eq:SDE_Y_tilde}) is 
\[\nabla f(f^{-1}(y))\sigma(f^{-1}(y))\sigma(f^{-1}(y))^\intercal \nabla f(f^{-1}(y))^\intercal = \begin{pmatrix}
\diag(\alpha_{1,11}\,\ldots,\alpha_{m,mm}) & \mathbf{0} \\
\mathbf{0} & a_{\mathcal{JJ}} + \sum_{i=1}^m y_i^2 \alpha_{i,\mathcal{JJ}}
\end{pmatrix}.\]
Hence, in light of the assumption that $\alpha_{i,ii}>0$, $i=1,\ldots,m$ and $a_{\mathcal{JJ}}\succ 0$, the diffusion matrix of (\ref{eq:SDE_Y_tilde}) is \emph{uniformly elliptic}. It is well known that such diffusion processes admit a positive probability density; see, e.g., Theorem 3.3.4 of \citet{Davies89}. Since the mapping $f$ is bijective, we conclude that $\tilde X$ also admits a positive transition density, so $\tilde X^\Delta$ is $\varphi$-irreducible. This completes the proof in light of Lemma \ref{lemma:irreducible}.
\end{proof}

\begin{proposition}\label{prop:affine}
Under Assumptions \ref{asp:state_space} and \ref{asp:admissible}, $X$ is a stochastically continuous affine process. 
\end{proposition}

\begin{proof}

For $T>0$ and any purely imaginary vector $u\in i\R^d$, define $M(t)\coloneqq e^{\phi(T-t,u) + \psi(T-t,u)^\intercal X(t)}$, where $\phi:\R_+\times i\R^d\mapsto \C$ and $\psi(t,u):\R_+\times i\R^d\mapsto\C^d$ are functions that are differentiable with respect to $t$. Applying It\^o's formula,
\begin{align*}
& M(t) \\
 = &\; M(0) + \int_0^t M(s\text{-})\psi(T-s,u)^\intercal \sigma(X(s))\,\ud W(s) + \int_0^t M(s\text{-}) \int_{\calX} \left(e^{\psi(T-s,u)^\intercal z}-1\right) N(\ud s, \ud z)\\
 + & \int_0^t M(s\text{-})\left[-\partial_t\phi(T-s,u)-\partial_t\psi(T-s,u)^\intercal X(s) + \psi(T-s,u)^\intercal \mu(X(s\text{-}))\right]\,\ud s \\
+ &\; \frac{1}{2}\int_0^t M(s\text{-})\psi(T-s,u)^\intercal \sigma(X(s\text{-}))\sigma(X(s\text{-}))^\intercal \psi(T-s,u)\,\ud s \\
=&\; M(0) + \int_0^t M(s\text{-})\psi(T-s,u)^\intercal \sigma(X(s))\,\ud W(s) + \int_0^t M(s\text{-}) \int_{\calX} \left(e^{\psi(T-s,u)^\intercal z}-1\right) \tilde N(\ud s, \ud z)\\ 
 +& \int_0^t M(s\text{-}) \left[-\partial_t\phi(T-s,u) + \psi(T-s,u)^\intercal b +\frac{1}{2} \psi(T-s,u)^\intercal a \psi(T-s,u) \right]\ud s \\
 + &\int_0^t M(s\text{-}) \bigg[-\partial_t\psi(T-s,u)^\intercal X(s\text{-}) + \psi(T-s,u)^\intercal \beta X(s\text{-}) + \frac{1}{2} \sum_{i=1}^d \psi(s,u)^\intercal\alpha_i \psi(s,u) X_i(s\text{-})\bigg]   \ud s  \\
 + &\int_0^t M(s\text{-}) (\lambda + \kappa^\intercal X(s\text{-})) \int_{\calX} \left(e^{\psi(T-s,u)^\intercal z}-1\right)\nu(\ud z) \ud s.
\end{align*}
Hence, if $\phi$ and $\psi$ satisfy the following generalized Riccati equations
\begin{align*}
\partial_t\phi(t,u) & = \psi(t,u)^\intercal b +\frac{1}{2} \psi(t,u)^\intercal a \psi(t,u) + \lambda \int_{\calX} \left(e^{\psi(t,u)^\intercal z}-1\right)\nu(\ud z),\\
\partial_t\psi_i(t,u) & = \psi(t,u)^\intercal \beta_i + \frac{1}{2} \sum_{i=1}^d \psi(t,u)^\intercal\alpha_i\psi(t,u) + \kappa_i \int_{\calX} \left(e^{\psi(t,u)^\intercal z}-1\right)\nu(\ud z),\quad i=1,\ldots,d,
\end{align*}
with $\phi(0,u) = 0 $ and $\psi(0,u) =u$, where $\beta_i$ is the $i$-th column of $\beta$, then 
\begin{equation}\label{eq:exp_HG}
M(t) = M(0) + \int_0^t M(s\text{-})\psi(T-s,u)^\intercal \sigma(X(s))\,\ud W(s) + \int_0^t M(s\text{-}) \int_{\calX} \left(e^{\psi(T-s,u)^\intercal z}-1\right) \tilde N(\ud s, \ud z).
\end{equation}

It follows from Proposition 6.1 and Proposition 6.4 of \citet{DuffieFilipovicSchachermayer2003} that under Assumption \ref{asp:admissible}, the preceding generalized Riccati equations have a unique solution $(\phi(\cdot,u),\psi(\cdot,u)): \R_+ \mapsto \C_- \times \C_-^m \times i\R^{d-m} $ for all $u\in \C_-^m \times i\R^{d-m}$, where $\C_-^m = \{z\in\C^m| \mathrm{Re}(z)\in \R_-^m\}$. Hence, 
\begin{equation}\label{eq:nonpositive_real_part}
\phi(t,u) + \psi(t,u)^\intercal x \in \C_-,\quad x\in\calX,
\end{equation}
under Assumption \ref{asp:state_space}. Further, Proposition 7.4 of \citet{DuffieFilipovicSchachermayer2003} asserts that 
\begin{equation}\label{eq:semiflow}
\begin{array}{l}
\phi(t+s, u)  = \phi(t,u) + \phi(s, \psi(t, u))\\
\psi(t+s, u)  = \psi(t, \psi(s, u))
\end{array}
\end{equation}
for all $t,s\in \R_+$ and $u\in \C_-^m \times i\R^{d-m}$.

In light of (\ref{eq:exp_HG}) and (\ref{eq:nonpositive_real_part}), $(M(t):0\leq t \leq T)$ is a local martingale with $\verti{M(t)} \leq 1$ for all $t$, thereby a martingale. So 
\begin{equation}\label{eq:expon_affine_ch}
\E_x[e^{u^\intercal X(T)}] = \E_x[M(T)] = \E_x[M(0)] = e^{\phi(T,u) + \psi(T,u)^\intercal x}, 
\end{equation}
namely the characteristic function $\E_x[e^{u^\intercal X(t)}]$ is exponential-affine in $x$. In addition, it is easy to verify via  (\ref{eq:semiflow}) and (\ref{eq:expon_affine_ch}) the Chapman–Kolmogorov equation 
\[\pr_x(X(t+s)\in \cdot) = \int_{\calX} \pr_x(X(t)\in\ud y)\pr_y(X(s)\in \cdot), \]
implying that $X$ is a time-homogeneous Markov process, thereby an affine process by (\ref{eq:expon_affine_ch}). 

At last, $\E_x[e^{u^\intercal X(t)}]$ is clearly continuous in $t$ by (\ref{eq:expon_affine_ch}), indicating that $X$ is stochastically continuous.
\end{proof} 

\begin{proof}[Proof of Theorem \ref{theo:state_ind_jump}\ref{item:ergo_ind}.]

We first show that $X$ is Harris recurrent. Theorem 1.1 of \citet{MeynTweedie93} asserts that $X$ is Harris recurrent if (i) $X$ is a \emph{Borel right process} \citep[p.55]{Getoor75}, and (ii) there exists a \emph{petite} set $K$ for $X$, such that $\pr_x(\tau_K<\infty)=1$ for all $x\in\calX$, where  $\tau_K=\inf\{t\geq 0:X(t)\in K\}$. 

For condition (i), we note that $X$ is a Feller process by Theorem 5.1 of \citet{Keller-ResselSchachermayerTeichmann11}, Proposition 8.2 of \citet{DuffieFilipovicSchachermayer2003}, and Proposition \ref{prop:affine}. The Feller property of $X$ trivially implies that $X$ is a Borel right process. 

For condition (ii), fix an arbitrary $\Delta>0$ and note that $X^\Delta$ is a Feller chain since $X$ is a Feller process. By Theorem 3.4 of \citet{MeynTweedie92_part1}, the Feller property of $X^\Delta$ and Proposition \ref{prop:irreducible} immediately imply that all compact sets are petite for $X^\Delta$, thereby petite for $X$. In the sequel, we will show that there exists a compact set $K$ such that $\pr_x(\tau_K<\infty)=1$ for all $x\in\calX$. To that send, we first establish the following Lyapunov inequality
\begin{equation}\label{eq:Lyapunov_inequality_ergodic}
\mathscr Ag(x)\leq -c_1 + c_2\ind_K(x),\quad x\in \calX,
\end{equation}
for some compact set $K$ and some positive finite constants $c_1$ and $c_2$, where $g(x)= \log(1+\vertii{x}_H^2)$ for some $d\times d$ matrix $ H\succ 0$. 

Since $\beta$ is a stable matrix, there exists a $d\times d$ matrix $ H\succ 0$ for which $-(H\beta+\beta^\intercal H)\succ 0$; see \citet[Theorem 2.3(G), p.134]{BermanPlemmons94}. It is straightforward to calculate the gradient and Hessian of $g(x)$ as follows 
\begin{align*}
\nabla g(x) = \frac{2Hx}{1+\vertii{x}_H^2}\quad\mbox{and}\quad 
\nabla^2 g(x) = \frac{2(1+\vertii{x}_H^2)H - 4Hx x^\intercal H}{(1+\vertii{x}_H^2)^2}.
\end{align*}
Hence,
\begin{equation}
\label{eq:GV}
\mathscr Gg(x)= \frac{2}{1+\vertii{x}_H^2}\left[x^\intercal H(b+\beta x) + \frac{1}{2}\sum_{i,j=1}^d \left(a_{ij} + \sum_{k=1}^d\alpha_{k,ij} x_k\right)\left(H-\frac{2Hxx^\intercal H}{1+\vertii{x}_H^2}\right)_{ij} \right].
\end{equation}
We note that for any $i,j=1,\ldots,d$, 
\[\verti{(Hxx^\intercal H)_{ij}} = \verti{(Hx)_i (Hx)_j} \leq \vertii{Hx}^2 \leq \vertiii{H}^2\vertii{x}^2\leq \ubar{\delta}^{-1}\vertiii{H}^2  \vertii{x}_H^2,\]
where the last inequality follows from (\ref{eq:norm_equivalence}). Hence, 
\begin{equation}
\label{eq:estimate_4}
\frac{|(Hxx^\intercal H)_{ij}|}{1+\vertii{x}_H^2} = O(1),
\end{equation}
as $\vertii{x}_H\to\infty$. Therefore, we can rewrite (\ref{eq:GV}) as
\[\mathscr Gg(x)= \frac{2x^\intercal H\beta x+O(\vertii{x}_H)}{1+\vertii{x}_H^2} 
=\frac{2x^\intercal H\beta x}{\vertii{x}_H(1+\vertii{x}_H)}+o(1),\]
as $\vertii{x}_H\to\infty$. Moreover, by virtue of (\ref{eq:norm_equivalence}) and the fact that $-(H\beta + \beta^\intercal H)\succ 0$, 
\[-2x^\intercal H\beta x = -x^\intercal (H\beta + \beta^\intercal H) x\geq \ubar{\gamma}\vertii{x}^2 \geq \ubar{\gamma}\bar{\delta}^{-1}\vertii{x}^2_H,\]
where $\ubar{\gamma}>0$ is the smallest eigenvalue of $-(H\beta + \beta^\intercal H)$. Therefore, 
\begin{equation}\label{eq:limsup_1}
\limsup_{\vertii{x}_H\to\infty} \mathscr Gg(x) = \limsup_{\vertii{x}_H\to\infty} \frac{2x^\intercal H\beta x}{\vertii{x}_H(1+\vertii{x}_H)} \leq -\ubar{\gamma}\bar{\delta}^{-1}.
\end{equation}

On the other hand, it is easy to see that 
$1 + (\vertii{x}_H+\vertii{z}_H)^2 \leq 2(1+\vertii{x}_H^2)(1+\vertii{z}_H^2)$ for all $x,z\in\R^d$. Thus,
\begin{equation}\label{eq:reverse_fatou}
\log\left(\frac{1+\vertii{x+z}_H^2}{1+\vertii{x}_H^2}\right)  \leq 
 \log\left(\frac{1+(\vertii{x}_H+\vertii{z}_H)^2}{1+\vertii{x}_H^2}\right)  
\leq \log (2 (1+\vertii{z}_H^2)). 
\end{equation}
It is easy to see that $\log (2 (1+\vertii{z}_H^2))$ is integrable on $\calX$, since $\E\log(1+\vertii{Z}_H)<\infty$ if and only if $\E\log(1+\vertii{Z})<\infty$ in light of (\ref{eq:norm_equivalence}). Then, we move the left-hand-side of \eqref{eq:reverse_fatou} to the right-hand-side and apply Fatou's lemma to  obtain 
\begin{equation*}
    \limsup_{\vertii{x}_H\to\infty}\int_{\calX}\log\left(\frac{1+\vertii{x+z}_H^2}{1+\vertii{x}_H^2}\right)\nu(\ud z) 
    \leq \int_{\calX}\limsup_{\vertii{x}_H\to\infty}\log\left(\frac{1+\vertii{x+z}_H^2}{1+\vertii{x}_H^2}\right)\nu(\ud z) = 0.
\end{equation*}
With $\kappa=\mathbf 0$,
\begin{equation}\label{eq:limsup_2}
\limsup_{\vertii{x}_H\to\infty}\mathscr Lg(x) = \limsup_{\vertii{x}_H\to\infty} 
 \lambda \int_{\calX} \log\left(\frac{1+(\vertii{x}_H+\vertii{z}_H)^2}{1+\vertii{x}_H^2}\right)\nu(\ud z) 
 \leq  0.   
\end{equation}
We then conclude from (\ref{eq:limsup_1}) and (\ref{eq:limsup_2}) that there exists $k>0$ for which 
\[\mathscr Ag(x) = \mathscr Gg(x) + \mathscr Lg(x)\leq -\frac{1}{2} \ubar{\gamma}\bar{\delta}^{-1},\] 
for all $x\in \calX$ with $\vertii{x}_H>k$. Then, it is easy to check that the inequality  (\ref{eq:Lyapunov_inequality_ergodic}) holds by setting $K=\{x\in \calX:\vertii{x}_H\leq k\}$, $c_1=\ubar{\gamma}\bar{\delta}^{-1}/2$, and $c_2=\max\{1, \sup_{x\in K}(\mathscr Ag(x) + c_1)\}$.

We are now ready to show $\pr_x(\tau_K<\infty)=1$ for all $x\in\calX$.  Define $T_n=\inf\{t\geq 0: |X(t)|>n\}$. It follows from (\ref{eq:ito_2}) and  (\ref{eq:Lyapunov_inequality_ergodic}) that  
\begin{equation}\label{eq:Dynkin_martingale}
g(X(t\wedge T_n)) \leq g(X(0)) + \int_0^{t\wedge T_n}\left[-c_1 + c_2\ind_K(X(s\text{-}))\right]\ud s + S_1(t\wedge T_n) + S_2(t\wedge T_n), \quad n\geq 1.
\end{equation}
Noting that $|X(t\text{-})|\leq n$ is bounded for $t\in[0,T_n)$, $(S_i(t\wedge T_n):t\geq 0)$ is a martingale, $i=1,2$. Then by the optional sampling theorem (see, e.g., \citet[p.19]{KaratzasShreve91})
\[\E_x[g(X(t\wedge \tau_K\wedge T_n))] \leq g(x) - c_1\E_x(t\wedge\tau_K\wedge T_n ) , \quad x\in \calX\setminus K, \; n\geq 1. \]
Therefore,
\[c_1\E_x(t\wedge\tau_K\wedge T_n ) \leq g(x), \quad x\in \calX\setminus K, \;n\geq 1,  \]
since $g(x)\geq 0$ for all $x\in\calX$. Note that $X$ is non-explosive, so $T_n\to\infty$ as $n\to\infty$ $\pr_x$-a.s. for all $x\in\calX$. Therefore, by sending $n\to\infty$ and then sending $t\to\infty$, we conclude from the monotone convergence theorem that $c_1\E_x(\tau_K) \leq  g(x)$ for $x\in\calX\setminus K$. Hence, $\pr_x(\tau_K<\infty)=1$ for all $x\in\calX$. Consequently, $X$ is Harris recurrent by Theorem 1.1 of \citet{MeynTweedie93}.

Theorem 1.2 of \citet{MeynTweedie93} states that given the Harris recurrence, $X$ is positive Harris recurrent if  $\sup_{x\in K} \E_x(\tau_K(\Delta)) < \infty$. We now show this is indeed the case. For any $\Delta>0$, let $\tau_K(\Delta) \coloneqq \Delta + \Theta^{\Delta}\circ\tau_K$ be the first hitting time on $K$ after $\Delta$, where $\Theta^{\Delta}$ is the \emph{shift operator}; see \citet[p.8]{Sharpe88}. Then, 
\begin{equation}
\label{eq:hitting_time}
\E_x(\tau_K(\Delta)-\Delta) = \int_{\calX}\pr_x(X(\Delta)\in \ud y)\E_y(\tau_K) \leq \int_{\calX}c_1^{-1}g(y)\pr_x(X(\Delta)\in \ud y) = c_1^{-1}\E_xg(X(\Delta)),
\end{equation}
for all $x\in\calX$. In addition, it follows from (\ref{eq:Dynkin_martingale}) that
\[\E_xg(X(\Delta\wedge T_n)) \leq g(x) + (c_2- c_1)\E_x(\Delta \wedge T_n),\quad x\in\calX,\;n\geq 1.\]
Then, by Fatou's lemma and the monotone convergence theorem, 
\begin{equation}
\label{eq:hitting_time2}
\E_xg(X(\Delta)) \leq \liminf_{n\to\infty}\E_xg(X(\Delta\wedge T_n)) \leq  g(x) + (c_2- c_1)\Delta,\quad x\in\calX.
\end{equation}
Combining (\ref{eq:hitting_time}) and (\ref{eq:hitting_time2}) yields that 
\[\E_x(\tau_K(\Delta)) \leq c_1^{-1}(g(x)+d),\quad x\in\calX. \]
Hence, $\sup_{x\in K} \E_x(\tau_K(\Delta)) < \infty$, which implies that $X$ is positive Harris recurrent by Theorem 1.2 of \citet{MeynTweedie93}.

Finally, Theorem 6.1 of \citet{MeynTweedie93_part2} asserts that if $X^\Delta$ is $\varphi$-irreducible, which is true by Proposition \ref{prop:irreducible}, then  a positive Harris recurrent process is ergodic, i.e. (\ref{eq:ergodic_convergence}) holds. 
\end{proof}

\begin{proof}[Proof of Theorem \ref{theo:state_dep_jump}\ref{item:ergo_dep}.]
Following the proof Theorem \ref{theo:state_ind_jump}\ref{item:ergo_ind}, it suffices to show the Lyapunov inequality  (\ref{eq:Lyapunov_inequality_ergodic}) holds under the assumptions of Theorem \ref{theo:state_dep_jump}\ref{item:ergo_dep}. In fact, we prove the following stronger result 
\begin{equation}\label{eq:Lyapunov_inequality_exp_ergodic}
\mathscr Ag(x)\leq -c_1 g(x) + c_2\ind_K(x) ,\quad x\in \calX,
\end{equation}
for some compact set $K$ and some positive finite constants $c_1$ and $c_2$, where $g(x)=(1+\vertii{x}_H^2)^{p/2}$ for some $d\times d$ matrix $ H\succ 0$ and some constant $p\geq 1$.

Since $\E\vertii{Z}<\infty$, there exists $p\geq 1$ for which $\E\vertii{Z}^p<\infty$. Since $\beta + \E(Z)\kappa^\intercal$ is stable, there exists a matrix $H\succ 0$ such that 
\begin{equation}\label{eq:H}
-[H (\beta + \E(Z)\kappa^\intercal)+ (\beta + \E(Z)\kappa^\intercal)^\intercal H]\succ 0. 
\end{equation}

It is straightforward to calculate the gradient and Hessian of $g(\cdot)$ as follows 
\[
\nabla g(x) = \frac{pg(x)}{1+\vertii{x}_H^2} Hx\quad\mbox{and}\quad
\nabla^2 g(x) = \frac{pg(x)}{1+\vertii{x}_H^2}\left[H + \frac{(p-2)Hxx^\intercal H}{1+\vertii{x}_H^2}\right].
\]
It then follows from (\ref{eq:estimate_4}) that as $\vertii{x}_H\to\infty$,
\begin{align}
\mathscr G g(x) &  = \frac{pg(x)}{1+\vertii{x}_H^2} \left[x^\intercal H(b+\beta x)+ \frac{1}{2}\sum_{i,j=1}^d(a_{i,j}+\sum_{k=1}^d\alpha_{k,ij}x_k)\left(H + \frac{(p-2)Hxx^\intercal H}{1+\vertii{x}_H^2}\right)_{i,j}\right] \nonumber\\
& = pg(x)\left(\frac{x^\intercal H\beta x}{\vertii{x}_H^2}+o(1)\right).\label{eq:estimate_2}
\end{align}

To analyze the asymptotic behavior of $\mathscr Lg(x)$, we apply the mean value theorem, namely
\begin{equation*}
\label{eq:Taylor}
g(x+ z )-g(x) =   \nabla g(\xi) ^\intercal  z
=  p(1+\vertii{\xi}_H^2)^{p/2-1} \xi^\intercal H z,
\end{equation*}
where $\xi=x+u  z $ for some $u\in(0,1)$. Note that $\vertii{\xi}_H$ lies between $\vertii{x}_H$ and $\vertii{x+ z}_H$ and  $\xi^\intercal H z \kappa^\intercal x$ lies between $x^\intercal H  z \kappa^\intercal x$ and
$(x+ z )^\intercal H z  \kappa^\intercal x$. It then follows that
\begin{equation}\label{eq:ratio2}
\frac{\kappa^\intercal x(g(x+ z )-g(x))}{g(x)} = p\cdot\frac{(1+\vertii{\xi}_H^2)^{p/2-1}}{(1+\vertii{x}_H^2)^{p/2}}\cdot \xi^\intercal H z  \kappa^\intercal x \sim p\cdot \frac{x^\intercal H  z \kappa^\intercal x}{\vertii{x}_H^2}
\end{equation}
as $\vertii{x}_H\to\infty$ for all $z\in\R^d$.\footnote{Here, we use the notation that $f(x)\sim g(x)$ if $\lim_{\vertii{x}_H\to\infty} \frac{f(x)}{g(x)}=1.$} Moreover, 
\begin{align}
\verti{g(x+ z )-g(x)} &=  
p(1+\vertii{\xi}_H^2)^{p/2-1} \verti{z^\intercal H \xi } \nonumber\\[0.5ex]
& \leq p(1+\vertii{\xi}_H^2)^{p/2-1} \vertii{z} \vertii{H\xi}\nonumber\\[0.5ex]
&  \leq p \ubar{\delta}^{-1}(1+\vertii{\xi}_H^2)^{p/2-1} \vertii{z}_H \vertiii{H}_H\vertii{\xi}_H \nonumber\\[0.5ex]
& \leq   p \ubar{\delta}^{-1}(1+\vertii{\xi}_H^2)^{p/2-1/2} \vertii{z}_H \vertiii{H}_H,\label{eq:ratio0}
\end{align}
where the second inequality follows from (\ref{eq:norm_equivalence}). So 
\begin{align*}
\verti{\frac{\kappa^\intercal x(g(x+ z )-g(x))}{g(x)}}&  \leq \frac{\verti{\kappa}\verti{x} p \ubar{\delta}^{-1}(1+\vertii{\xi}_H^2)^{p/2-1/2} \vertii{z}_H \vertiii{H}_H}{(1+\vertii{x}_H^2)^{p/2}}\nonumber\\
&\leq  p \ubar{\delta}^{-2}\vertiii{H}_H \vertii{\kappa}_H \vertii{z}_H\frac{(1+\vertii{\xi}_H^2)^{p/2-1/2}}{(1+\vertii{x}_H^2)^{p/2-1/2}},
\end{align*}
where the second inequality follows from (\ref{eq:norm_equivalence}). Note that
\begin{equation*}\label{eq:norm_bound}
1+\vertii{\xi}_H^2  = 1+\vertii{x+uz}_H^2 \leq 2(1+\vertii{x}_H^2)(1+\vertii{uz}_H^2) \leq 2(1+\vertii{x}_H^2)(1+\vertii{z}_H^2),
\end{equation*}
so 
\begin{equation}\label{eq:integrable}
\int_{\calX}\left|\frac{\kappa^\intercal x(g(x+ z )-g(x))}{g(x)} \right| \nu(\ud z) \leq 2^{p/2-1/2}p\ubar{\delta}^{-2}\vertiii{H}_H\vertii{\kappa}_H \int_{\calX}(1+\vertii{z}_H^2)^{p/2}\,\nu(\ud z) < \infty.
\end{equation}
By (\ref{eq:ratio2}) and (\ref{eq:integrable}), the dominated convergence theorem dictates that 
\[
\kappa^\intercal x\int_{\calX} (g(x+ z )-g(x))\,\nu(\ud z)\sim
pg(x)\cdot \int \frac{x^\intercal H  z \kappa^\intercal
  x}{\vertii{x}_H^2} \,\nu(\ud z)
= pg(x)\cdot\frac{x^\intercal H \E (Z)\kappa^\intercal x}{\vertii{x}_H^2},
\]
and thus
\begin{equation}\label{eq:estiamte_new}
\mathscr Lg(x) =(\lambda+\kappa^\intercal x) \int_{\calX} (g(x+ z )-g(x))\,\nu(\ud z)\sim pg(x)\frac{x^\intercal H \E (Z)\kappa^\intercal x}{\vertii{x}_H^2},
\end{equation}
as $\vertii{x}_H\to\infty$. Combining (\ref{eq:estimate_2}) and (\ref{eq:estiamte_new}), 
\[\mathscr Ag(x) = \mathscr Gg(x)+\mathscr Lg(x)=pg(x)\left(\frac{x^\intercal H(\beta + \E(Z)\kappa^\intercal) x}{\vertii{x}_H^2}+o(1)\right)\]
as $\vertii{x}_H\to\infty$. By  (\ref{eq:H}), the definition of the matrix $H$,
\begin{align*}
- x^\intercal H(\beta + \E(Z)\kappa^\intercal) x = & -\frac{1}{2} x^\intercal \left[H(\beta + \E(Z)\kappa^\intercal) + (\beta + \E(Z)\kappa^\intercal)^\intercal H\right]x \geq  \frac{1}{2}
\ubar{\gamma} \vertii{x}^2 \geq \frac{1}{2} \ubar{\gamma}\bar{\delta}^{-1}\vertii{x}^2_H,   
\end{align*}
where $\ubar{\gamma}>0$ is the smallest eigenvalue of $-\left[H(\beta + \E(Z)\kappa^\intercal) + (\beta + \E(Z)\kappa^\intercal)^\intercal H\right]$. Hence, there exists $k > 0$ for which
$\mathscr Ag(x)\leq -\frac{1}{4}p \ubar{\gamma}\bar{\delta}^{-1} g(x)$
for all $x\in\calX$ with $\vertii{x}_H>k$. Therefore, (\ref{eq:Lyapunov_inequality_exp_ergodic}) holds by setting $K=\{x\in \calX:\vertii{x}_H\leq k\}$, $c_1=p\ubar{\gamma}\bar{\delta}^{-1}/4$, and $c_2=\max\{1,\sup_{x\in K}(\mathscr Ag(x) + c_1g(x))\}$. 
\end{proof}

\subsection{Exponential Ergodicity}\label{sec:exp_ergodic}

\begin{proof}[Proof of Theorem \ref{theo:state_ind_jump}\ref{item:exp_ergo_ind}.] 
Note that if $\E\vertii{Z}^p<\infty$ for some $p>0$, then $\E\vertii{Z}^q<\infty$ for all $q\in(0,p]$. We assume that $p\in(0,1)$, because the case that $p\geq 1$ is covered by Theorem \ref{theo:state_dep_jump}\ref{item:exp_ergo_dep}, which will be proved later.

Since $\beta$ is stable, there exists a matrix $H\succ 0$ such that $-(H \beta + \beta^\intercal H)\succ 0$. We show that $g_q(x)=(1+\vertii{x}_H^2)^{q/2}$ satisfies the inequality (\ref{eq:Lyapunov_inequality_exp_ergodic}) for some compact set $K$ and some positive finite constants $c_1,c_2$. Note that 
\[
g_q(x+z)-g_q(x)\leq (1+\vertii{x}_H^2+\vertii{z}_H^2)^{q/2} - (1+\vertii{x}_H^2)^{q/2} 
=\frac{q}{2}\xi^{q/2-1}\vertii{z}_H^q,
\]
where the equality follows from the mean value theorem and $\xi\in(1+\vertii{x}_H^2, 1+\vertii{x}_H^2+\vertii{z}_H^2)$.  Since $\xi>1$ and $p\in(0,1)$, we have $g_q(x+z)-g_q(x) \leq \frac{q}{2}\vertii{z}_H^q$. Likewise, it can be shown that $g_q(x) - g_q(x+z)\leq \frac{q}{2}\vertii{z}_H^q$. Hence, $|g_q(x+z)-g_q(x)| \leq \frac{q}{2}\vertii{z}_H^q$ and 
\[\verti{\int_{\calX} g_q(x+ z )-g_q(x)\,\nu(\ud z)}  \leq  \int_{\calX} \verti{g_q(x+ z )-g_q(x)}\,\nu(\ud z)
\leq \frac{q}{2}\E\vertii{Z}_H^q <\infty,
\]
It follows that with $\kappa=\mathbf 0$,
\[
\mathscr Lg_q(x) = \lambda \int_{\calX} (g_q(x+ z )-g_q(x))\,\nu(\ud z) = O(1),
\]
as $\vertii{x}_H\to\infty$. Moreover, applying (\ref{eq:estimate_2}) to $g_q(x)$, 
\begin{equation*}\label{eq:estiamte_new_3}
\mathscr Ag_q(x) = \mathscr Gg_q(x)+\mathscr Lg_q(x)=qg_q(x)\left(\frac{x^\intercal H\beta x}{\vertii{x}_H^2}+o(1)\right),
\end{equation*}
as $\vertii{x}_H\to\infty$. By the definition of the matrix $H$,
\begin{align*}
- x^\intercal H\beta x = & -\frac{1}{2} x^\intercal  (H\beta + \beta^\intercal H)x \geq  \frac{1}{2}
\ubar{\gamma} \vertii{x}^2 \geq \frac{1}{2} \ubar{\gamma}\bar{\delta}^{-1}\vertii{x}^2_H,   
\end{align*}
where $\ubar{\gamma}>0$ is the smallest eigenvalue of $-(H\beta + \beta^\intercal H)$. Hence, there exists $k > 0$ such that $\mathscr Ag(x)\leq -\frac{1}{4}p \ubar{\gamma}\bar{\delta}^{-1} g(x)$ 
for all $x\in\calX$ with $\vertii{x}_H>k$. Therefore, 
\begin{equation}\label{eq:Lyapunov_inequality_exp_ergodic_new}
    \mathscr Ag_q(x)\leq -c_1 g_q(x) + c_2\ind_K(x) ,\quad x\in \calX,
\end{equation}
where   $K=\{x\in \calX:\vertii{x}_H\leq k\}$, $c_1=p\ubar{\gamma}\bar{\delta}^{-1}/4$, and $c_2=\max\{1,\sup_{x\in K}(\mathscr Ag(x) + c_1g(x))\}$.

We apply It\^o's formula to $e^{c_1t}g_q(X(t))$. In particular, by (\ref{eq:ito_2}),
\begin{align*}
e^{c_1t}g_q(X(t)) =& \,g_q(X(0)) +\int_0^t e^{c_1s}[c_1g_q(X(s\text{-})) + \mathscr Ag_q(X(s\text{-}))]\ud s  \\ 
&+ \int_0^t e^{c_1s} \nabla g_q(X(s\text{-}))^\intercal\sigma(X(s))\,\ud W(s) \\ 
& 
+ \int_0^t e^{c_1s}\int_{\calX} (g_q(X(s\text{-})+z) - g_q(X(s\text{-})))\tilde N(\ud s, \ud z). 
\end{align*}
Clearly, the two stochastic integrals above are both martingales up to time $T_n$, where $T_n=\{t\geq 0: |X(t)|>n\}$. It follows from (\ref{eq:Lyapunov_inequality_exp_ergodic_new}) and the optional sampling theorem that
\[e^{c_1t}\E_x g_q(X(t\wedge T_n)) \leq g_q(x) + \E_x\int_0^{t\wedge T_n} e^{c_1s}\cdot c_2\ind_K(X(s))\,\ud s \leq g_q(x) + c_2c_1^{-1} \E_x e^{t\wedge T_n}.\]
We now apply Fatou's lemma and the monotone convergence theorem to conclude that 
\begin{equation}\label{eq:exp_ergo_ineq}
e^{c_1t}\E_xg_q(X(t)) \leq g_q(x) + c_2c_1^{-1}\cdot \liminf_{n\to\infty} \E_x e^{t\wedge T_n} =g_q(x) + c_2c_1^{-1} e^{c_1t}.
\end{equation}
Then we can adopt the argument used in the proof of Theorem 6.1 of \citet{MeynTweedie93_part3} to conclude that because of  (\ref{eq:exp_ergo_ineq}),  there exist positive finite constants $d_q$ and $\rho_q$ such that
\[\vertii{\pr_x(X(t)\in\cdot) - \pi(\cdot)}_{g_q+1}\leq d_q(g_q(x)+1)e^{-\rho_qt},\quad t\geq 0,\;x\in\calX.\]
By (\ref{eq:norm_equivalence}), there exist positive constants $d_1$ and $d_2$ such that 
$d_1\leq \verti{\frac{f_q(x)}{g_q(x)+1}} \leq d_2$ for all $x\in\calX$.
Hence, 
\[\vertii{\pr_x(X(t)\in\cdot) - \pi(\cdot)}_{f_q}\leq c_q f_q(x)e^{-\rho_qt},\quad t\geq 0,\;x\in\calX,\]
where $c_q=d_qd_2/d_1$. 
\end{proof}

\begin{proof}[Proof of Theorem \ref{theo:state_dep_jump}\ref{item:exp_ergo_dep}.]
Following the proof of Theorem \ref{theo:state_ind_jump}\ref{item:exp_ergo_ind}, it suffices to show that (\ref{eq:Lyapunov_inequality_exp_ergodic_new}) holds under the present assumptions. Note that $\E\vertii{Z}^q<\infty$ for all $q\in[1, p]$ since $\E\vertii{Z}^p<\infty$. Hence, we can apply the Lyapunov inequality (\ref{eq:Lyapunov_inequality_exp_ergodic}) to $g_q(x)$, which results in  (\ref{eq:Lyapunov_inequality_exp_ergodic_new}).  
\end{proof}

\subsection{Remarks on the Strong Mean-Reversion Condition}\label{sec:remarks}

The key condition that we impose to establish positive Harris recurrence of $X$ is the strong mean-reversion condition, i.e, $\beta+\E(Z)\kappa^\intercal$ is a stable matrix. Indeed, this condition cannot be relaxed in general as illustrated by the following example.

\begin{proposition}\label{prop:transient}
Suppose that $d=1$, $m=1$, and Assumptions \ref{asp:state_space}--\ref{asp:irreducible} hold. If $\E\verti{Z}<\infty$ and $\beta +\E(Z)\kappa>0$, then $X$ is transient .
\end{proposition}
\begin{proof}
The proof also relies on Lyapunov inequalities; see Theorem 3.3 of \citet{StramerTweedie94}. Specifically, transience follows if there exists a  bounded function $g$ and a closed set $K$ such that
\begin{equation}
\label{eq:Lyapunov_inequality_transient}
\mathscr Ag(x)\geq 0,\quad x\in \calX \setminus K, 
\end{equation}
and
\begin{equation}
\label{eq:Lyapunov_inequality_transient_2}
\sup_{x\in K}g(x) < g(x_0),\quad x_0\in \calX \setminus K.
\end{equation}

Let $g(x)=1-e^{-\epsilon x}$ for some $\epsilon>0$. Obviously, $g$ is bounded for $x\in \calX=\R_+$.  Then, 
\begin{align*}
\mathscr A g(x) & = (b+\beta x)g'(x) + \frac{1}{2}(a+\alpha x) g''(x) + (\lambda + \kappa x) \int_{\R_+}(g(x+z)-g(x))\,\nu(\ud z) \\
 & =  e^{-\epsilon x}\left[\epsilon(b+\beta x) - \frac{\epsilon^2}{2}(a+\alpha x)+(\lambda+\kappa x) \int_{\R_+}(1-e^{-\epsilon z})\,\nu(\ud z)\right] \\
 & = e^{-\epsilon x}\left[\left(\epsilon \beta - \frac{\epsilon^2}{2}\alpha + \kappa(1-\E e^{-\epsilon Z}) \right)x +\epsilon b - \frac{\epsilon^2}{2}a +\lambda (1-\E e^{-\epsilon Z})\right].
\end{align*}
Let $h(\epsilon)$ be the coefficient of $x$ in the brackets above, i.e.,
$h(\epsilon) \coloneqq \epsilon \beta - \frac{\epsilon^2}{2}\alpha + \kappa(1-\E e^{-\epsilon Z})$. 
Clearly, $h(0)=0$ and $h'(0) = \beta + \kappa \E(Z) >0$, yielding that $h(\epsilon)>0$ for some $\epsilon>0$. Fixing this $\epsilon$, we see that 
$\mathscr Ag(x) \sim e^{-\epsilon x} h(\epsilon)x$ as $x\to\infty$. Hence, there exists $k>0$ such that $\mathscr Ag(x) >0$ for $x \in \calX\setminus K$, where $K\coloneqq [0, k]$, proving (\ref{eq:Lyapunov_inequality_transient}). Moreover, (\ref{eq:Lyapunov_inequality_transient_2}) is true since  $g(x)$ is increasing in $x$. 
\end{proof}
 
The ``boundary'' case, i.e. $\beta+\E(Z)\kappa^\intercal=0$, is more complicated as the behavior of the process may depend on other parameters. We leave its analysis for future research.

\section{Limit Theorems}\label{sec:limits}

In this section, we prove SLLNs and FCLTs for additive functionals of $X$ of the form $\int_0^t h(X(s))\,\ud s$ or $\sum_{i=1}^nh(X(i\Delta))$ for some function $h$. Limit theorems for both discrete-time and continuous-time Markov processes have been extensively studied in the past; see, e.g., \citet{GlynnMeyn1996}, \citet{KontoyiannisMeyn03}, \citet[chap.17]{MCSS09}, and references therein. In particular, positive Harris recurrence is ``almost'' sufficient for a LLN to hold. Conditions for FCLTs, on the other hand, often include exponential ergodicity, or Lyapunov inequalities of the form similar to (\ref{eq:Lyapunov_inequality_exp_ergodic}). 

Nevertheless, existing FCLTs for discrete-time Markov processes are not applicable to the skeleton chain $X^\Delta$ because they typically require one to establish a ``discrete-time'' version of the Lyapunov inequality of the form $\E_x[g(X(\Delta))]\leq cg(x)$ for some constant $c<1$, some function $g\geq 1$, and all $x$ off a compact set. This is awkward mathematically given the fact that the transition measure $\pr_x(X(\Delta)\in\cdot)$ is not known explicitly. Our approach to establish (\ref{eq:FCLT_disc_func}) is to first consider the scenario in which $X(0)$ follows the stationary distribution. We then apply an FCLT for stationary sequences, i.e., Theorem 3.1 of \citet[p.351]{EthierKurtz86}, whose conditions can be verified as a consequence of exponential ergodicity (\ref{eq:exp_ergodic_convergence}). To generalize the FCLT to an arbitrary initial state we follow an argument similar to one used in \cite{GlynnMeyn1996}.

The asymptotic variances, $\sigma_h^2$ in (\ref{eq:FCLT_cont_func}) and $\gamma_h^2$ in (\ref{eq:FCLT_disc_func}), can be expressed in terms of the solution to a \emph{Poisson equation}; see, e.g., \citet{GlynnMeyn1996}. But it typically has no closed form in terms of the parameters $(a,\alpha_1\ldots,\alpha_d,b,\beta,\lambda,\kappa,\nu)$ of the SDE \eqref{eq:SDE}. However, when $h$ is the (vector-valued) identity function, we are indeed able to analytically derive both the asymptotic mean and asymptotic covariance matrix that appear in the corresponding FCLT (see Corollary \ref{cor:limits}), thanks to the tractable affine structure. 

\subsection{Strong Law of Large Numbers}\label{sec:SLLN}

\begin{proof}[Proof of Theorem \ref{theo:state_ind_jump}\ref{item:SLLN_ind} and Theorem \ref{theo:state_dep_jump}\ref{item:SLLN_dep}.]

We have established positive Harris recurrence and ergodicity of $X$ in Section \ref{sec:ergodicity} under the assumptions of Theorem \ref{theo:state_ind_jump}\ref{item:SLLN_ind} or Theorem \ref{theo:state_dep_jump}\ref{item:SLLN_dep}. So $\pi(|h|)<\infty$ for any measurable function $h:\calX\mapsto\R$ with $\vertii{h}_{f_p}<\infty$. The SLLN (\ref{eq:LLN_cont_func}) then follows  from Theorem 2 of \citet{Sigman90}. 

For the skeleton chain $X^\Delta$, note that the stationary distribution $\pi$ of $X$ is necessarily invariant for $X^\Delta$. In addition, $X^\Delta$ is $\varphi$-irreducible by Proposition \ref{prop:irreducible}, so $X^\Delta$ is positive Harris recurrent. Hence, the SLLN (\ref{eq:LLN_disc_func}) follows from Theorem 17.1.7 of \citet[p.427]{MCSS09}. 
\end{proof}

\subsection{Functional Central Limit Theorem}\label{sec:CLT} 

\begin{proof}[Proof of Theorem \ref{theo:state_ind_jump}\ref{item:FCLT_ind} and Theorem \ref{theo:state_dep_jump}\ref{item:FCLT_dep}.] 

Fix $q>2$ and an arbitrary measurable function $h:\calX\mapsto\R$ with $\vertii{h^q}_{f_p}<\infty$. 

We have shown in Section \ref{sec:exp_ergodic} that there exists a matrix $H\succ 0$, a compact set $K$ and positive finite constants $c_1,c_2$ such that
$\mathscr Ag(x)\leq -c_1g(x) + c_2\ind_K(x)$ for all $x\in\calX$, 
where  $g(x) = (1+\vertii{x}_H^2)^{p/2}$. Thanks to (\ref{eq:norm_equivalence}), $\vertii{h}_{f_p}<\infty$ if and only if $\vertii{h}_g<\infty$. Moreover, we have shown in Section \ref{sec:ergodicity} that $K$ is a petite set for $X$. It then follows immediately from Theorem 4.4 of \citet{GlynnMeyn1996} that  (\ref{eq:FCLT_cont_func}) holds as $n\to\infty$ $\pr_x$-weakly in $\mathcal D[0,1]$ for all $x\in\calX$. 

We now show that (\ref{eq:FCLT_disc_func}) holds $\pr_\pi$-weakly in $\mathcal D[0,1]$, where $\pi$ is the stationary distribution of $X$. This can be done by applying an FCLT for stationary sequences to $\{\bar{h}(X(n\Delta):n=0,1,\ldots\}$, which is a mean-zero stationary sequence if $X(0)\sim\pi$, where $\bar{h}(x)\coloneqq h(x)-\pi(h)$.  

Specifically, let $\mathscr F_k$ and $\mathscr F^k$ denote the $\sigma$-algebras generated by $(X(n\Delta):n\leq k)$ and  $(X(n\Delta):n\geq k)$, respectively. Let $\varphi_1(l) \coloneqq \sup_{\Gamma\in\mathscr F^{k+l}}\E_\pi \verti{\pr(\Gamma|\mathscr F_k) - \pr(\Gamma)}$ denote the \emph{measure of mixing} \citep[p.346]{EthierKurtz86} of $\mathscr F_k$ and $\mathscr F^{k+l}$ associated with the $L^1$-norm. Then, by Theorem 3.1 and Remark 3.2(b) of \citet[p.351]{EthierKurtz86}, it suffices to verify that for some $\epsilon>0$,
\begin{equation}\label{eq:FCLT_conditions}
\E_\pi\left[\verti{\bar{h}(X(n\Delta))}^{2+\epsilon}\right]<\infty\quad\mbox{and}\quad \sum_{l=0}^\infty[\varphi_1(l)]^{\epsilon/(2+\epsilon)}<\infty.
\end{equation}

Let $\epsilon=q-2>0$. Then, 
\[\E_\pi\left[\verti{\bar{h}(X(n\Delta))}^{2+\epsilon}\right] =\pi(\bar{h}^q) \leq \vertii{\bar{h}^q}_{f_p} \pi(f_p) <\infty,\] 
verifying the first condition in \eqref{eq:FCLT_conditions}. To verify the second, note that by the Markov property, for any $\Gamma\in\mathscr F^{k+l}$ there exists a function $w_\Gamma$ with $\verti{w_\Gamma(\cdot)}\leq 1$ such that $\pr[\Gamma|\mathscr F_{k+l}]=w_\Gamma(X((k+l)\Delta))$.  If $X(0)\sim\pi$, then for any $\Gamma\in\mathscr F^{k+l}$,
\begin{align*}
\verti{\pr(\Gamma|\mathscr F_k) - \pr(\Gamma)} 
& = \verti{\E[w_\Gamma(X((k+l)\Delta))|\mathscr F_k]-\E[w_\Gamma(X((k+l)\Delta))]} \nonumber\\[0.5ex]
& = \verti{\int_{\calX} w_\Gamma(y)\pr_{X(k\Delta)}(X(l\Delta)\in\ud y) - \int_{\calX}\int_{\calX}w_\Gamma(y)\pr_x(X((k+l)\Delta)\in\ud y)\pi(\ud x) } \nonumber\\[0.5ex]
& \leq \vertii{\pr_{X(k\Delta)}(X(l\Delta)\in\cdot) -\pi(\cdot)}, \label{eq:mixing}
\end{align*}
where $\vertii{\cdot}$ is the total variation norm, where the inequality follows from Definition \ref{def:norm} and the fact that $\verti{w_\Gamma(\cdot)}\leq 1$. It follows that 
\[
\varphi_1(l)
 \leq  \E_\pi \vertii{\pr_{X(k\Delta)}(X(l\Delta)\in\cdot) -\pi(\cdot)}
\leq  c_p e^{-\rho_p l\Delta }\E_\pi[f(X(k\Delta))]  =  c_p \pi(f_p) e^{-\rho_p l\Delta},
\]
where the second inequality holds because of Theorem \ref{theo:state_ind_jump}\ref{item:exp_ergo_ind} and Theorem \ref{theo:state_dep_jump}\ref{item:exp_ergo_dep}. This immediately implies $\sum_{l=0}^\infty[\varphi_1(l)]^{\epsilon/(2+\epsilon)}<\infty$. Therefore, we conclude that (\ref{eq:FCLT_disc_func}) holds $\pr_\pi$-weakly in $\mathcal D[0,1]$. 

We now prove that (\ref{eq:FCLT_disc_func}) indeed holds as $n\to\infty$ $\pr_x$-weakly in $\mathcal D[0,1]$ for all $x\in\calX$. To that end, we first show that
\begin{equation}\label{eq:uniform_convergence}
\pr_x\left(\lim_{n\to\infty}\sup_{0\leq t\leq 1} \verti{Y_n(t)-Y_{n,l}(t)}=0\right)=1,\quad x\in\calX,
\end{equation}
for any positive integer $l$, where $Y_n(t) \coloneqq  n^{-1/2}\sum_{i=1}^{\lfloor nt \rfloor} \bar{h}(X(i\Delta))$ and $Y_{n,l}(t) \coloneqq n^{-1/2}\sum_{i=l+1}^{\lfloor nt \rfloor+l} \bar{h}(X(i\Delta))$.
Note that for all sufficiently large $n$, 
\begin{align*}
\sup_{0\leq t\leq 1} \verti{Y_n(t)-Y_{n,l}(t)} & = \frac{1}{n}\verti{\sum_{i=1}^n \bar{h}(X(i\Delta))-\sum_{i=l+1}^{n+l} \bar{h}(X(i\Delta))}^2\\ 
&= \frac{1}{n}\verti{\sum_{i=1}^l \bar{h}(X(i\Delta))-\sum_{i=n+1}^{n+l} \bar{h}(X(i\Delta))}^2\\[0.5ex] 
 & \leq  \frac{1}{n}\sum_{i=1}^l \bar{h}^2(X(i\Delta)) + \frac{1}{n}\sum_{i=n+1}^{n+l} \bar{h}^2(X(i\Delta)) \to 0,\quad \pr_x-\mathrm{a.s.},
\end{align*}
as $n\to\infty$ for all $x\in\calX$, because 
$\frac{1}{n}\sum_{i=1}^n \bar{h}^2(X(i\Delta)) \to \pi(\bar{h}^2)<\infty$, $\pr_x-\mathrm{a.s.}$, 
as $n\to\infty$ for all $x\in\calX$, thanks to Theorem \ref{theo:state_ind_jump}\ref{item:SLLN_ind} and Theorem \ref{theo:state_dep_jump}\ref{item:SLLN_dep}. This completes the proof of (\ref{eq:uniform_convergence}). 

Let $\phi$ be a bounded continuous functional $\phi$ on $\mathcal D[0,1]$. Then (\ref{eq:uniform_convergence}) implies that for any positive integer $l$, $\verti{\E_x[\phi(Y_n)]-\E_x[\phi(Y_{n,l})]}\to 0$ as $n\to\infty$ for all $x\in\calX$. This limit can be rewritten as 
\begin{equation}\label{eq:norm_bound_1}
\lim_{n\to\infty}\verti{\E_x[\phi(Y_n)]- \int_{\calX}\pr_x(X(l\Delta)\in \ud y)\E_y[\phi(Y_{n,l})]} = 0,\quad x\in\calX.
\end{equation}
On the other hand, note that
\[
\verti{\int_{\calX}\pr_x(X(l\Delta)\in \ud y)\E_y[\phi(Y_{n,l})] - \E_\pi[\phi(Y_n)]} \leq  \vertii{\pr_x(X(l\Delta)\in \cdot)-\pi(\cdot)}  \cdot \sup_{g\in\mathcal D[0,1]} \verti{\phi(g)}.
\]
Since $\vertii{\pr_x(X(l\Delta)\in \cdot)-\pi(\cdot)} \to 0$ as $l\to\infty$ by Theorem \ref{theo:state_ind_jump}\ref{item:ergo_ind} and Theorem \ref{theo:state_dep_jump}\ref{item:ergo_dep}, for any $\delta>0$ we can choose $l$ so large that 
\begin{equation}\label{eq:norm_bound_2}
\verti{\int_{\calX}\pr_x(X(l\Delta)\in \ud y)\E_y[\phi(Y_{n,l})] - \E_\pi[\phi(Y_n)]} \leq \delta.
\end{equation}
It then follows from (\ref{eq:norm_bound_1}) and (\ref{eq:norm_bound_2}) that 
$\limsup_{n\to\infty}\verti{\E_x[\phi(Y_n)]-\E_\pi[\phi(Y_n)]}\leq \delta$.
Since (\ref{eq:FCLT_disc_func}) holds $\pr_\pi$-weakly in $\mathcal D[0,1]$, we must have 
$\lim_{n\to\infty}\verti{\E_\pi[\phi(Y_n)]-\E_\pi[\phi(W)]}=0$, 
and thus 
\[\limsup_{n\to\infty}\verti{\E_x[\phi(Y_n)]-\E_\pi[\phi(W)]}\leq \delta.\]
Sending $\delta\to 0$ yields that (\ref{eq:FCLT_disc_func}) holds $\pr_x$-weakly in $\mathcal D[0,1]$ for all $x\in\calX$.
\end{proof}

\subsection{A Special Case} Thanks to the affine structure, the asymptotic mean and the asymptotic variance can be derived analytically when $h$ is the identity function, i.e. $h(x)=x$. Note that with $h$ being $\R^d$-valued, the corresponding SLLN and FCLT are multivariate. The calculation follows closely the approach used in \cite{ZhangBlanchetGieseckeGlynn15} so we omit the details. 

\begin{corollary}\label{cor:limits}
If Assumptions \ref{asp:state_space}--\ref{asp:irreducible} hold and $\E\vertii{Z}<\infty$, then 
\[
\pr_x\left(\lim_{t\to\infty}\frac{1}{t}\int_0^t h(X(s))\,\ud s = v \right)=1, \quad x\in\calX,
\]
where $v=-(\beta + \E(Z)\kappa^\intercal )^{-1} (b+\lambda \E(Z))$. Furthermore, if $\E\vertii{Z}^{2+\epsilon}<\infty$ for some $\epsilon>0$, then 
\[n^{1/2}\left(\frac{1}{n}\int_0^{n\cdot} X(s)\,\ud s - v\right)\Rightarrow \Sigma^{1/2} W(\cdot),
\]
as $n\to\infty$ $\pr_x$-weakly in $\mathcal D_{\R^d}[0,1]$ for all $x\in\calX$, where 
\[\Sigma=A(a+\lambda\E (Z Z^\intercal))A^\intercal  +\sum_{i=1}^m  v_i A(\alpha_i+\kappa_i\E (Z Z^\intercal) )A^\intercal.
\]
\end{corollary}

\section*{Acknowledgments}
The first author was partially supported by the Hong Kong Research Grants Council under General Research Fund (ECS 624112). The second author gratefully acknowledges the support and the intellectual environment of the Institute for Advanced Study at the City University of Hong Kong, where this work was completed.

\bibliographystyle{chicago}
\bibliography{AJD}

\end{document}